\documentclass[copyright,creativecommons]{eptcs}
\providecommand{\event}{TTATT 2013} 

\usepackage{doc,makeidx,url,amssymb,amsmath,amsthm,graphicx,paralist}
\usepackage[T1]{fontenc}
\usepackage[strings]{underscore} 

\usepackage{color}

\title{Random Generation of \\
  Nondeterministic Finite-State Tree Automata}
\def\titlerunning{Random Generation of FTAs}

\author{Thomas Hanneforth
\institute{Universit\"at Potsdam, Department Linguistik \\
  Karl-Liebknecht-Stra\ss e 24--25, 14476 Potsdam, Germany}
  \email{thomas.hanneforth@uni-potsdam.de}
\and Andreas Maletti\footnotemark[1]
  \qquad\qquad Daniel Quernheim\footnotemark[1]
 \institute{Universit\"at Stuttgart, Institut f\"ur Maschinelle
  Sprachverarbeitung \\ Pfaffenwaldring~5b, 70569 Stuttgart, Germany}
  \email{\{maletti,daniel\}@ims.uni-stuttgart.de}}
\def\authorrunning{Hanneforth, Maletti, and Quernheim}

\def\publicationstatus{}
\def\copyrightholders{\authorrunning}

\allowdisplaybreaks

\newtheorem{definition}{Definition}
\newtheorem{example}[definition]{Example}
\newtheorem{theorem}[definition]{Theorem}

\providecommand{\abs}[1]{\ensuremath{\lvert #1 \rvert}}
\providecommand{\nat}[0]{\ensuremath{{\mathbb N}}}
\providecommand{\seq}[3]{\ensuremath{#1_{#2}, \dotsc, #1_{#3}}}

\DeclareMathOperator{\rk}{rk}

\begin{document}

\renewcommand{\thefootnote}{\fnsymbol{footnote}}
\footnotetext[1]{Financially supported by
    the German Research Foundation~(DFG) grant MA\,/\,4959\,/\,1-1.}
\renewcommand{\thefootnote}{\arabic{footnote}}

\maketitle

\begin{abstract}
  Algorithms for (nondeterministic) finite-state tree automata~(FTAs)
  are often tested on random FTAs, in which all internal transitions
  are equiprobable.  The run-time results obtained in this manner are
  usually overly optimistic as most such generated random FTAs are
  trivial in the sense that the number of states of an equivalent
  minimal deterministic FTA is extremely small.  It is demonstrated
  that nontrivial random FTAs are obtained only for a narrow band of
  transition probabilities.  Moreover, an analytic analysis yields a
  formula to approximate the transition probability that yields the
  most complex random FTAs, which should be used in experiments.
\end{abstract}

\section{Introduction}
Nondeterministic finite-state tree automata~(FTAs) play a major role in
several areas of natural language processing.  For example, the
\textsc{Berkeley} parser~\cite{petbarthikle06} uses a (weighted) FTA
as do syntax-based approaches to statistical machine translation.
Toolkits~\cite{maykni06,lensimvoj12} for FTAs allow users to easily run
experiments.  However, algorithms like determinization typically
cannot be tested on real-world examples (like the FTA of the
\textsc{Berkeley} parser) due to their complexity.  In such cases, the
inputs are often random FTAs, which are typically created by fixing
densities, which are the probabilities that any given potential
transition (of a certain group) is indeed a transition of the
generated FTA (see~\cite{zij97} and \cite{les95} for the generation of
random finite-state string automata).

It is known~\cite{chahanparzia04} that for string automata most random
automata are trivial in the sense that the equivalent minimal
deterministic finite-state string automaton is extremely small.  Here,
we observe the same effect for FTAs, which means that testing
algorithms on random FTAs also has to be done carefully to avoid
vastly underestimating their actual run-time.  To simplify such
experiments and to make them more representative we provide both
empirical and analytic evidence for the nontrivial (difficult)
cases.  In the empirical evaluation we create many random FTAs with a
given number~$n$ of (useful) states using a given transition
density~$d$.  We then determinize and minimize them, and we record the
number of states of the resulting canonical FTA (i.e., the equivalent
minimal deterministic FTA).  As expected, outside a narrow density
band the randomly generated FTAs yield very small canonical~FTA, which
means that they are in a sense trivial.\footnote{Trivial here does not
  mean that the recognized tree language is uninteresting, but rather
  it only relates to its complexity.}  It can be observed that if the
density is above the upper limit of the band, then the FTAs accept
almost everything, whereas FTAs with densities below the lower limit
accept almost nothing.  This observation also justifies calling them
trivial.

In the analytic evaluation, given a number~$n$ of states, we compute
densities~$d(n)$, for which we expect the most difficult FTAs.
Looking at the empirical results, the formula predicts the narrow
density band belonging to complex random FTAs very well.
Consequently, we promote experiments with random FTAs that use exactly
these predicted densities in order to avoid experiments with (only)
trivial FTAs.  Finally, we discuss how parameter changes affect our
results.  For example, adding another binary input symbol does not
move the interesting narrow density band, but it generally does
increase the sizes of the obtained canonical FTA.  Moreover, it shows
that whenever we obtain large deterministic FTA after determinization
(i.e., before minimization), then the corresponding canonical FTA are
also large (i.e., after minimization).  This demonstrates that our
implementation of determinization is rather efficient.

\section{Finite tree automata}
\label{sec:fta}
The power set of a set~$S$ is $\mathcal{P}(S) = \{S' \mid S' \subseteq
S\}$.  The set of nonnegative integers is denoted by~$\nat$.  An
alphabet is simply a finite set of symbols.  A ranked alphabet
$(\Sigma, \mathord{\rk})$ consists of an alphabet~$\Sigma$ and a
mapping $\mathord{\rk} \colon \Sigma \to \nat$, which assigns a rank
to each symbol.  For every $k \in \nat$, we let $\Sigma_k = \{ \sigma
\in \Sigma \mid \rk(\sigma) = k\}$ be the set of symbols of rank~$k$.
We also write $\sigma^{(k)}$ to indicate that the symbol~$\sigma$ has
rank~$\rk(\sigma) = k$.  To keep the presentation simple, we typically
write just~$\Sigma$ for the ranked alphabet~$(\Sigma, \mathord{\rk})$
and assume that the ranking~`$\mathord{\rk}$' is clear from the
context.  Moreover, we often drop obvious universal quantifications
like `$k \in \nat$' in expressions like~`for all $k \in \nat$ and
$\sigma \in \Sigma_k$'.  Our trees have node labels taken from a
ranked alphabet~$\Sigma$ and leaves can also be labeled by elements of
a finite set~$Q$.  The rank of a symbol~$\sigma \in \Sigma$ determines
the number of direct children of all nodes labeled~$\sigma$.  Given a
set~$T$, we write $\Sigma(T)$ for the set $\{\sigma(\seq t1k) \mid
\sigma \in \Sigma_k, \seq t1k \in T\}$.  The set~$T_\Sigma(Q)$ of
\mbox{$\Sigma$-trees indexed by~$Q$} is defined as the smallest
set~$T$ such that $\Sigma(T) \cup Q \subseteq T$.  We write~$T_\Sigma$
for~$T_\Sigma(\emptyset)$.

Next, we recall finite-state tree automata~\cite{gecste84,gecste97}
and the required standard constructions.  In general, finite-state
tree automata~(FTAs) offer an efficient representation of the regular
tree languages.  We distinguish a (bottom-up) deterministic variant
called deterministic FTA, which will be used in our size measurements.
A \emph{finite-state tree automaton}~(FTA) is a system $(Q, \Sigma, F,
P)$, where (i)~$Q$~is a finite set of \emph{states}, (ii)~$\Sigma$~is
a ranked alphabet of \emph{input symbols}, (iii)~$F \subseteq Q$ is a
set of \emph{final} states, and (iv)~$P \subseteq \Sigma(Q) \times Q$
is a finite set of \emph{transitions}.  It is \emph{deterministic} if
for every $t \in \Sigma(Q)$ there exists at most one~$q \in Q$ such
that $(t, q) \in P$.  We often write a transition~$(t, q) \in
\Sigma(Q) \times Q$ as~$t \to q$.  The size of the FTA~$M = (Q,
\Sigma, F, P)$ is~$\abs M = \abs Q$.  This is arguably a crude measure
for the `size', but it will mostly be used for deterministic FTAs,
where it is commonly used.

\begin{example} 
  \label{ex:fta}
  As illustration we consider the FTA $M_{\text{ex}} = (\{0, 1, 2,
  3\}, \Sigma, \{3\}, P)$, where $\Sigma = \{\sigma^{(2)},
  \alpha^{(0)}\}$ and $P$~contains the following transitions.
  Clearly, this FTA is not deterministic.
  \begin{align*}
    \alpha &\to 0 & \alpha &\to 2 & \sigma(0, 0) &\to 1 &
    \sigma(1,0) &\to 1 & \sigma(1, 2) &\to 3 & \sigma(1,3) &\to 3
  \end{align*}
\end{example}

In the following, let $M = (Q, \Sigma, F, P)$ be an FTA.  For every
$\sigma \in \Sigma_k$, let $\overline \sigma \colon \mathcal{P}(Q)^k
\to \mathcal{P}(Q)$ be such that $\overline\sigma(\seq Q1k) = \{ q \in
Q \mid \forall 1 \leq i \leq k,\, \exists q_i \in Q_i \colon
\sigma(\seq q1k) \to q \in P\}$ for all $\seq Q1k \subseteq Q$.  Next,
we define the action of the transitions on a tree~$t \in T_\Sigma(Q)$.
Let $P \colon T_\Sigma(Q) \to \mathcal{P}(Q)$ be such that $P(q) =
\{q\}$ for every $q \in Q$ and $P(\sigma(\seq t1k)) = \overline
\sigma(P(t_1), \dotsc, P(t_k))$ for every $\sigma \in \Sigma_k$ and
$\seq t1k \in T_\Sigma$.  The FTA~$M$ accepts the tree
language~$L(M)$, which is given by $L(M) = \{ t \in T_\Sigma \mid P(t)
\cap F \neq \emptyset\}$.  Two FTAs $M_1$~and~$M_2$ are equivalent if
$L(M_1) = L(M_2)$.  For example, for the FTA of Example~\ref{ex:fta}
we have $P(\alpha) = \{0, 2\}$ and $P(\sigma(\alpha, \alpha)) =
\{1\}$.  Moreover, $\sigma(\sigma(\sigma(\alpha, \alpha), \alpha),
\alpha) \in L(M_{\text{ex}})$.  For our experiments, we generate a
random FTA, then determinize it and compute the number of states of
the canonical FTA, which the equivalent minimal deterministic FTA.
For an FTA~$M = (Q, \Sigma, F, P)$, we construct the deterministic FTA
$\mathcal{P}(M) = (\mathcal{P}(Q), \Sigma, F', P')$ such that $F' = \{
Q' \subseteq Q \mid Q' \cap F \neq \emptyset\}$ and
\[ P' = \{ \sigma(\seq Q1k) \to \overline \sigma(\seq Q1k) \mid
\sigma \in \Sigma_k, \seq Q1k \subseteq Q\} \enspace. \]

\begin{theorem}[see~\protect{\cite[Theorem~1.10]{don70}}]
  \label{thm:det}
  $\mathcal{P}(M)$ is a deterministic FTA that is equivalent to~$M$.
\end{theorem}

\begin{example}
  \label{ex:det}
  For the FTA~$M_{\text{ex}}$ of Example~\ref{ex:fta} an equivalent
  deterministic FTA is~$\mathcal{P}(M_{\text{ex}})$, which is given by
  $\mathcal{P}(M_{\text{ex}}) = (\mathcal{P}(Q), \Sigma, F',
  P')$ with $F' = \{3, 13, 23, 123\}$ and $P'$~contains the
  (non-trivial) transitions\footnote{We abbreviate sets
    like~$\{0,2\}$ to just~$02$.}
  \begin{align*}
    \alpha &\to 02 & \sigma(02, 02) &\to 1 &
    \sigma(1, 02) &\to 13 & \sigma(13, 02) &\to 13 &
    \sigma(1, 13) &\to 3 & \sigma(1, 3) &\to 3 \enspace.
  \end{align*}
\end{example}

A deterministic FTA is \emph{minimal} if there is no strictly smaller
equivalent deterministic FTA.  For example, the deterministic FTA in
Example~\ref{ex:det} is minimal.  It is known~\cite{bra68} that for
every deterministic FTA we can compute an equivalent minimal
deterministic FTA, called the \emph{canonical} FTA. 

\section{Random generation of FTAs}
\label{sec:random}
First, we describe how we generate random FTAs.\footnote{These FTAs
  shall serve as test inputs for algorithms that operate on FTAs such
  as determinization, bisimulation minimization, etc.  Naturally, the
  size of an equivalent minimal FTA would be an obvious complexity
  measure for them, but it is \textsc{PSpace}-complete to determine
  it, and the size of the canonical FTA is naturally always
  bigger, so trivial FTAs according to our measure are also trivial
  under the minimal FTA size measure.}  We closely follow the random
generation outlined in~\cite{zij97}, augmented by density parameters
$d_2$~and~$d_0$, which is similar to the setup of~\cite{les95}.  Both
methods~\cite{zij97,les95} are discussed in~\cite{chahanparzia04},
where they are applied to finite-state string automata~\cite{yu97}.
For an event~$E$, let $\pi(E)$~be the probability of~$E$.  To keep the
presentation simple, we assume that the ranked alphabet~$\Sigma$ of
input symbols is binary (i.e., $\Sigma = \Sigma_2 \cup
\Sigma_0$).\footnote{Our approach can easily be adjusted to
  accommodate non-binary ranked alphabets.  We can imagine a model in
  which only one density~$d$ governs all transitions, but this model
  requires a slightly more difficult analytic analysis.}  Note that
all regular tree languages~\cite{gecste84,gecste97} can be encoded
using a binary ranked alphabet.  In order to randomly generate an FTA
$M = (Q, \Sigma, F, P)$ with $n = \abs Q$~states and binary and
nullary transition densities $d_2$~and~$d_0$, each transition (incl.\
the target state) is a random variable and each state is a random
variable representing whether it is final or not.  More precisely, we
use the following approach:
\begin{compactitem}
\item $Q = \{1, \dotsc, n\}$,
\item $\pi(q \in F) = \tfrac 12$ for all $q \in Q$ (i.e., for each
  state~$q$ the probability that it is final is~$\tfrac 12$),
\item $\pi(\alpha \to q \in P) = d_0$ for all nullary $\alpha \in
  \Sigma_0$ and $q \in Q$, and
\item $\pi(\sigma(q_1, q_2) \to q \in P) = d_2$ for all binary
  symbols~$\sigma \in \Sigma_2$ and all states $q_1, q_2, q \in Q$.
\end{compactitem}
If the such created FTA is not trim\footnote{The FTA~$M$ is trim if
  for every $q \in Q$ there exist $t \in T_\Sigma$ and $t' \in
  T_\Sigma(\{q\}) \setminus T_\Sigma$ such that (i)~$q \in P(t)$ and
  (ii)~$P(t') \cap F \neq \emptyset$.},  
then we start over and generate a new FTA.  Thus, all our randomly
generated FTAs indeed have $n$~useful states.

\begin{table}[t]
  \centering
  \begin{tabular}{r|r|rrcr|r|rr}
    $n$ & $d_2$ & $d'_2$ & conf.\ interval &
    \mbox{\qquad} & $n$ & $d_2$ & $d'_2$  & conf.\ interval \\* 
    \cline{1-4} \cline{6-9}
    2 & .6364 & .6264 & [.5769,.6804] & & 8 & .0431 & .0408 & [.0317,.0526] \\*
    3 & .2965 & .2570 & [.2091,.3159] & & 9 & .0341 & .0342 & [.0272,.0430] \\*
    4 & .1696 & .1334 & [.1024,.1737] & & 10 & .0276 & .0282 & [.0231,.0343] \\*
    5 & .1094 & .0855 & [.0642,.1138] & & 11 & .0228 & .0251 & [.0208,.0303] \\*
    6 & .0763 & .0635 & [.0475,.0848] & & 12 & .0192 & .0212 & [.0182,.0248] \\*
    7 & .0562 & .0501 & [.0380,.0662] & & 13 & .0164 & .0189 & [.0162,.0219] 
  \end{tabular}
  \caption{Expected density~$d_2$ (see
    Theorem~\protect{\ref{thm:main}}) and observed density~$d'_2$ (see
    Section~\protect{\ref{sec:data}}) for the most complex FTAs
    in Setting~(A).  We
    also report 
    confidence intervals for the confidence
    level $p > .95$.} 
  \label{tab:densities}
\end{table}

\section{Analytic analysis}
\label{sec:analysis}
In this section, we present a short analytic analysis and compute
densities, for which we expect the randomly generated FTAs to be
non-trivial.  More precisely, we estimate for which densities the
determinization (and subsequent minimization) returns the largest
canonical FTAs.  It is known from the generation of random finite-state
string automata~\cite{chahanparzia04} that the largest deterministic
automata are obtained during determinization if each state~$q \in Q$
occurs with probability~$\tfrac 12$ in the transition target of a
transition, in which the source states are selected uniformly at
random.  This observation was empirically confirmed multiple times by
independent research groups~\cite{les95,zij97,chahanparzia04}.  In
addition, all transition target states are equiprobable for input
states that are drawn uniformly at random.  This latter observation
supports the optimality claim by arguments from information theory
because if all target states of a transition are equiprobable, then
the entropy of the transition is maximal.  While these facts support
our hypothesis (and subsequent conclusions), we also rely on an
empirical evaluation in Section~\ref{sec:data} to validate it.

The determinization (see Section~\ref{sec:fta}) constructs the state
set~${\cal P}(Q)$.  Let ${\cal P}(M) = ({\cal P}(Q), \Sigma, F', P')$
be the deterministic FTA given the random FTA~$M$ with $n = \abs
Q$.\footnote{Note that the transitions of this DTA are random
  variables distributed according to the determinization construction
  of Section~\ref{sec:fta} applied to~$M$.}  According to our
intuition, the probability that a state~$Q' \in {\cal P}(Q)$ is the
transition target of a given transition~$\sigma(Q_1, Q_2)$, where
$Q_1$~and~$Q_2$ are uniformly selected at random from~${\cal P}(Q)$,
should be~$2^{-n}$ [i.e., $\pi(\sigma(Q_1, Q_2) \to Q' \in P') =
2^{-n}$].  Thus, in particular, each given state~$q \in Q$ is in the
(real) successor state~$Q''$ of the given transition~$\sigma(Q_1,
Q_2)$ with probability~$\tfrac 12$ [i.e., $\pi(q \in Q'') = \tfrac
12$].  Given $\sigma \in \Sigma_2$ and $q \in Q$, let $\pi_{\sigma,q}
= \pi(q \in Q'')$, where $Q_1$~and~$Q_2$ are uniformly selected at
random from~${\cal P}(Q)$ and $Q'' = \overline \sigma(Q_1, Q_2)$.
Similarly, given $\alpha \in \Sigma_0$ and $q \in Q$, let
$\pi_{\alpha,q} = \pi(q \in \overline \alpha)$.  It is easily seen
that $\pi_{\sigma, q} = \pi_{\sigma', q'}$ for all $\sigma, \sigma'
\in \Sigma_2$ and $q, q' \in Q$ because in our generation model all
$\sigma$-transitions in~$M$ with $\sigma \in \Sigma_2$ are
equiprobable.\footnote{Note that the individual $\sigma$-transitions
  of~${\cal P}(M)$ are not equiprobable.  For example, the transition
  $\sigma(\emptyset, \emptyset) \to Q$ is impossible.}  Thus, we
simply write~$\pi_2$ instead of~$\pi_{\sigma, q}$.  The same property
holds for nullary symbols, so we henceforth write~$\pi_0$
for~$\pi_{\alpha, q}$.  Moreover, as in the previous section, let
$n$~be the number of states of the original random FTA, and let
$d_2$~and~$d_0$ be the transition densities of it.  For $n = 1$ the
presented intuition cannot be met.

\begin{theorem}
  \label{thm:main}
  Let $n > 1$.  If $d_2 = 4(1 - \sqrt[n^2]{.5})$ and $d_0 = \tfrac
  12$, then $\pi_2 = \pi_0 = \tfrac 12$.
\end{theorem}

\begin{proof}
  We start with $\pi_0$.  Let $\alpha \in \Sigma_0$ and $q \in Q$.
  Then $\pi(q \in \overline \alpha) = \pi(\alpha \to q \in P) = d_0 =
  \tfrac 12$ as required.  For~$\pi_2$ let $Q_1, Q_2 \in {\cal P}(Q)$
  be selected uniformly at random, $\sigma \in \Sigma_2$, and $q \in
  Q$.  Then
  \begin{align*}
    \pi(q \in \overline \sigma(Q_1, Q_2)) &= 1 - \pi(q \notin
    \overline \sigma(Q_1, Q_2)) \\
    &= 1 - \prod_{q_1, q_2 \in
      Q} \Bigl(1 - \pi(q_1 \in Q_1) \cdot \pi(q_2 \in Q_2) \cdot
    \pi(\sigma(q_1, q_2) \to q \in P) \Bigr) \\ 
    &= 1 - \bigl( 1 - \frac{d_2}{4} \bigr)^{n^2} = 1 - \bigl(1 - 1 +
    \sqrt[n^2]{.5} \bigr)^{n^2} = 1 - (\sqrt[n^2]{.5})^{n^2} = \frac
    12 \enspace.
    \tag*{\qedhere}
    \end{align*}
\end{proof}

Table~\ref{tab:densities} lists some values~$d_2$ computed according
to Theorem~\ref{thm:main} for the sizes~$n \in \{2, \dotsc, 13\}$.

\section{Empirical analysis}
\label{sec:data}
In this section, we want to confirm that the computed densities indeed
represent the most difficult instances for the random FTAs constructed
in Section~\ref{sec:random}.  We use two settings:
\begin{compactitem}
\item[(A)] $\Sigma = \{\alpha^{(0)}, \sigma^{(2)}\}$ and
\item[(B)] $\Sigma = \{\alpha^{(0)}, \sigma^{(2)}, \delta^{(2)}\}$.
\end{compactitem}
For both settings (A)~and~(B) and varying densities~$d_2 = e^{\frac{x
    \cdot {\log D_n}}{20}}$ and $d_0 = \tfrac 12$, where $D_n = 4(1 -
\sqrt[n^2]{.5})$, for all $0 \leq x \leq 40$ and sizes $2 \leq n \leq
13$, we generated at least 40~trim FTAs.  The ratio of trim FTAs for
various densities and state set sizes can be found in
Table~\ref{tab:trim}.  Generally, larger state sets and higher 
densities increase the chance of obtaining a trim FTA.  The choice of
densities we made ensures that sufficiently many data points (in
equally-spaced steps on a logarithmic scale) will exist on both sides
of the density that is predicted to generate the most difficult
instances.  We will discuss why we favored the logarithmic scale over
a linear scale in the next paragraph.  These FTAs were subsequently
determinized, minimized, and the sizes of the canonical FTAs were
recorded.  These operations were performed inside our new tree
automata toolkit \textsc{TAlib}\footnote{Additional information about
  \textsc{TAlib} is available at
  \url{http://www.ims.uni-stuttgart.de/forschung/ressourcen/werkzeuge/talib.en.html}.}.

\begin{figure}[t]
  \centering
  \includegraphics[angle=0,width=.49\linewidth]{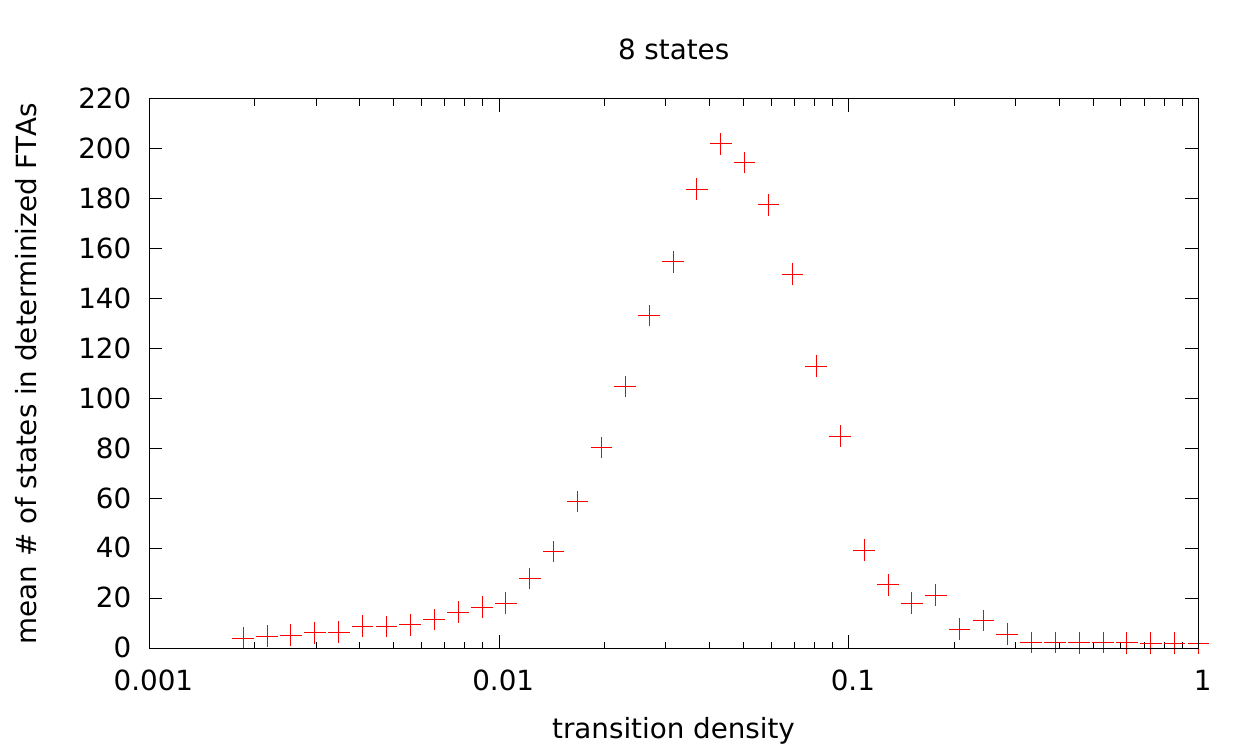}
  \includegraphics[angle=0,width=.49\linewidth]{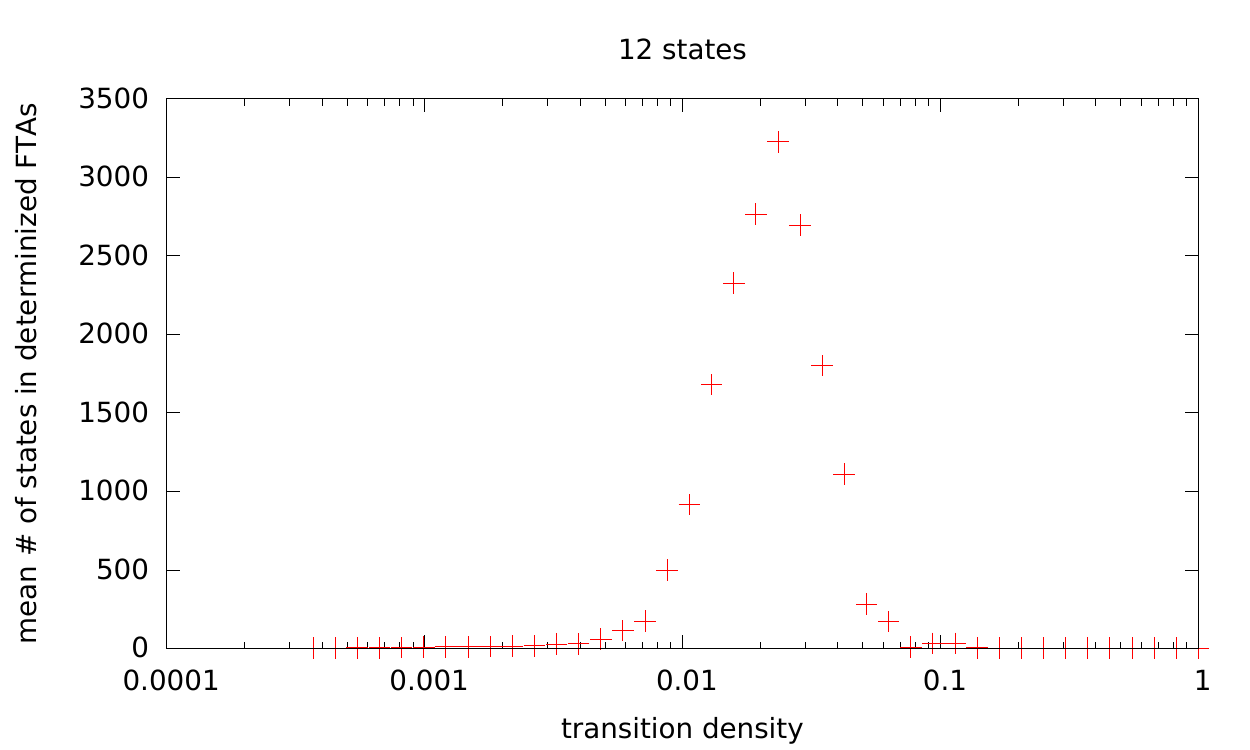}
  \caption{Graphs plotting the mean size of the determinized FTAs
    obtained by determinization over the density for
    FTAs of size $8$~and~$12$.  Minimization reduces the mean size,
    but does not move the peak.}
  \label{fig:graphs}
\end{figure}

Our experiments confirm the theoretical predictions.  A peak in the
mean size of the determinized FTAs can be observed where it is
predicted.  Exemplary graphs for setting~(A) and $n \in \{8, 12\}$ are
presented in Figure~\ref{fig:graphs} on a logarithmic scale.  Since
these graphs appear to be log-normal distributions\footnote{A
  log-normal distribution is a distribution of a random variable whose
  logarithm is normally distributed.  A log-normal distribution
  usually arises as the product of independent normal distributions.
  We leave the question of how exactly this distribution can be
  derived for further research.
}, we computed the mean and the variance of these log-normal
distributions, interpreting the density as the random variable and the
number of states as the frequency.  The relevant statistics are
reported in Table~\ref{tab:densities} for setting~(A).  All predicted
densities are in the confidence interval for the confidence level $p >
.95$ (that is, the predicted density is within $1.96 \cdot \sigma$
distance of the observed mean, where $\sigma$ is the standard
deviation).

It is worth noting that the location of the peak does not change
between settings (A)~and~(B).  This means that the size of the
alphabet of binary symbols does not influence the hardness of the
problem; the only difference is the size of the resulting determinized
FTAs, which is generally larger in setting~(B).  Also, minimization
does not change the location of the peak, which means that hard
instances for determinization are also hard instances for
minimization.  In addition, we also performed experiments that
confirmed the similar result of~\cite{chahanparzia04} for finite-state
string automata using the string automata toolkit FSM\textless
2.0\textgreater\ of~\cite{han10}.

\begin{table}[t]
  \centering 
  \begin{tabular}{r|r|r|r|r|r|r|r|r|r|r}
    $d_2\, \backslash\, n$ & 2 & 4 & 6 & 7 & 8 & 9 & 10 & 11 & 12 & 13
    \\ \hline
    .01 & & & & 7\% & 11\% & 18\% & 27\% & 38\% & 50\% & 64\% \\
    .05 & & & 68\% & 82\% & 92\% & 98\% & 99\% & 100\% & 100\% & 100\% \\
    .10 & & 54\% & 90\% & 96\% & 99\% & 99\% & 100\% & 100\% & 100\% &
    100\% \\
    .25 & & 83\% & 97\% & 98\% & 99\% & 100\% & 100\% & 100\% & 100\%
    & 100\% \\
    .50 & 47\% & 88\% & 96\% & 98\% & 100\% & 100\% & 100\% & 100\% &
    100\% & 100\% 
  \end{tabular}
  \caption{Ratio of trim FTAs in a set of randomly generated FTAs
    parameterized by density~$d_2$ (rows) and number~$n$ of states
    (columns).  Blank entries indicate that no such experiment was
    performed.} 
  \label{tab:trim}
\end{table}

\bibliographystyle{eptcs}
\bibliography{ttatt}

\end{document}